\documentclass[11pt, a4paper]{article}
\usepackage[centering]{geometry}
\usepackage[english]{babel}

\usepackage{amsmath,amsfonts,amsthm}
\usepackage{enumitem}
\usepackage{graphicx}
\usepackage{wrapfig}
\usepackage{cite}
\usepackage{xspace}
\usepackage{caption}
\usepackage{subcaption}
\graphicspath{{./figs/}}

\newcommand{\R}{\ensuremath{\mathbb{R}}\xspace}

\newcommand{\etal}{\emph{et~al.}\xspace}

\DeclareMathOperator{\conv}{conv}
\DeclareMathOperator{\NVD}{NVD}
\DeclareMathOperator{\FVD}{FVD}

\DeclareMathOperator{\VD}{VD}

\newtheorem{theorem}{Theorem}[section]
\newtheorem{lemma}[theorem]{Lemma}

\newtheorem{fact}[theorem]{Fact}
\newtheorem{corollary}[theorem]{Corollary}

\title{Improved Time-Space Trade-offs for Computing Voronoi Diagrams
  \thanks{MK~was supported in part by MEXT KAKENHI Nos.~17K12635, 15H02665, and 24106007. 
  BB, WM and PS were supported in part by DFG Grants MU 3501/1 and
  MU 3501/2. YS was supported by the DFG within the research 
  training group ``Methods for Discrete Structures'' (GRK 1408) 
  and by GIF Grant 1161.  
  AvR and MR were supported by JST ERATO Grant 
  Number JPMJER1201, Japan.
  A preliminary version appeared as
  B.~Banyassady, M.~Korman, W.~Mulzer, A.~van Renssen, 
  M.~Roeloffzen, P.~Seiferth, and Y.~Stein.
  \emph{Improved Time-Space Trade-offs for Computing Voronoi Diagrams}.
  Proc.~34th STACS, pp.~9:1--9:14, 2017.}
}
              
\author{
  Bahareh~Banyassady%
  \thanks{{Institut f\"ur Informatik, Freie Universit\"at Berlin, Berlin, Germany},
 \texttt{[bahareh, mulzer, pseiferth, yannikstein]@inf.fu-berlin.de}.}\,
 \and
  Matias~Korman%
  \thanks{{Tohoku University, Sendai, Japan},
  \texttt{mati@dais.is.tohoku.ac.jp}.}\,
  \and
  Wolfgang~Mulzer%
  \footnotemark[2]\,
  \and
  Andr\'e~van~Renssen%
  \thanks{{School of Information Technologies, University of Sydney, Sydney, Australia}, \newline
  \texttt{andre.vanrenssen@sydney.edu.au}.}\,
  \and
  Marcel~Roeloffzen%
  \thanks{{Department of Mathematics and Computer Science, TU Eindhoven, Eindhoven, the Netherlands},
  \texttt{m.j.m.roeloffzen@tue.nl}.}\,
  \and
  Paul~Seiferth%
  \footnotemark[2]\,
  \and
  Yannik~Stein%
  \footnotemark[2]\,
}

\date{}

\begin{document}

\maketitle

\begin{abstract}
Let $P$ be a planar set of $n$ \emph{sites} in general position. 
For $k \in \{1, \dots, n-1\}$, the Voronoi diagram 
\emph{of order $k$} for $P$ is obtained by subdividing the plane into
\emph{cells} such that points in the same cell have the same 
set of nearest $k$ neighbors in $P$.
The \emph{\textup{(}nearest site\textup{)} Voronoi diagram} ($\NVD$) and 
the 
\emph{farthest site Voronoi diagram} ($\FVD$) are the particular cases 
of $k=1$ and $k=n-1$, respectively. For any given $K \in \{1, \dots, 
n-1\}$, the family of all higher-order Voronoi diagrams of order 
$k = 1, \dots, K$ for $P$ can be computed in total time $O(nK^2+ n
\log n)$ using $O(K^2(n-K))$ space [Aggarwal \etal, DCG'89; Lee, TC'82]. 
Moreover, $\NVD$ and $\FVD$ for $P$ can be computed in $O(n\log n)$ 
time using $O(n)$ space [Preparata, Shamos, Springer'85].

For $s \in \{1, \dots, n\}$, an \emph{$s$-workspace} algorithm has 
random access to a read-only array with the sites of $P$ in arbitrary
order.  Additionally, the algorithm may use $O(s)$ words, of 
$\Theta(\log n)$ bits each, for reading and writing intermediate data.
The output can be written only once and cannot be accessed or 
modified afterwards.

We describe a deterministic $s$-workspace algorithm for computing 
$\NVD$ and $\FVD$ for $P$ that runs in $O((n^2/s)\log s)$ 
time. Moreover, we generalize our $s$-workspace 
algorithm so that for any given $K \in O(\sqrt{s})$, 
we compute the family of all higher-order Voronoi diagrams 
of order $k = 1, \dots, K$ for $P$ in total expected time 
$O\bigl(\frac{n^2 K^5}{s}(\log s + K \, 2^{O(\log^* K)}) \bigr)$ or
in total deterministic time
$O\bigl(\frac{n^2 K^5}{s}(\log s + K \log K) \bigr)$. 
Previously, for Voronoi diagrams, the only known $s$-workspace 
algorithm runs in \emph{expected} time
$O\bigl((n^2/s) \log s + n \log s \log^*s\bigr)$ [Korman \etal, WADS'15]
and only works for $\NVD$ (i.e., $k=1$). Unlike the previous 
algorithm, our new method is very simple and does not rely on 
advanced data structures or random sampling techniques.
\end{abstract}

\section{Introduction}
In recent years, we have seen an explosive growth of small 
distributed devices such as tracking devices and wireless sensors.
These gadgets are small, have only limited energy supply, are easily
moved, and should not be too expensive. To accommodate these needs,
the amount of memory on them is tightly budgeted. This poses a
significant challenge to software developers and algorithm designers:
how to create useful and efficient programs in the presence of strong
memory constraints?

Memory constraints have been studied since the introduction of 
computers (see for example Pohl~\cite{p-amsacm-69}). The first 
computers often had limited memory compared to the available 
processing power. As hardware progressed, this gap narrowed,
other concerns became more important, and the focus of algorithms 
research shifted away from memory-constrained models. However, 
nowadays, memory constraints are again an important problem to tackle
for these new devices as well as for huge datasets that have become 
available through cloud computing.

An easy way to model algorithms with memory constraints is to assume
that the input is stored in a read-only memory. This is appealing for
several reasons. From a practical viewpoint, writing to external
memory is often a costly operation, e.g., if the data resides on a
read-only medium such as a DVD or on hardware where writing is slow
and wears out the material, such as flash memory. Similarly, in 
concurrent environments, writing operations may lead to race 
conditions.  Thus, it is useful to limit or simply disallow writing
operations. From a theoretical viewpoint, this model is also 
advantageous: keeping the working memory separate from the 
(read-only) input memory allows for a more detailed accounting of the
space requirements of an algorithm and for a better understanding of
the required resources. In fact, this is exactly the approach taken
by computational complexity theory. Here, one defines complexity 
classes that model \emph{sublinear} space requirements, such as 
the complexity class of problems that use a logarithmic amount of 
space~\cite{AroraBa09}.

Some of the earliest results in this setting concern the sorting
problem~\cite{mp-ssls-80,mr-sromswmdm-96}.  Suppose we want to sort
data items whose total size is $n$ bits, all of them residing in a
read-only memory. For our computations, we can use a workspace of
$O(b)$ bits freely (both read and write operations are allowed). 
Then, it is known that the time-space product must be
$\Omega(n^2)$~\cite{bc-atstosgsmc-82}, and a matching upper bound for
the case $b\in \Omega(\log n) \cap O(n/\log n)$ was given by Pagter
and Rauhe~\cite{pr-otstofs-98} ($b$ is the available workspace in
bits). A result along these lines is known as a \emph{time-space
  trade-off}~\cite{s-mcepc-08}.

The model used in this work was introduced by 
Asano~\etal~\cite{amrw-cwagp-10}, following similar earlier
models~\cite{BronnimanChCh04,cc-mpga-07}. Asano~\etal provided
constant workspace algorithms for many classic problems from
computational geometry, such as computing convex hulls, Delaunay
triangulations, Euclidean minimum spanning trees, or shortest paths
in polygons~\cite{amrw-cwagp-10}. Since then, the model has enjoyed
increasing popularity, with work on shortest paths in
trees~\cite{AsanoMuWa11} and time-space trade-offs for computing
shortest paths~\cite{abbkmrs-mcasp-11,Har-Peled16}, visibility
regions in simple polygons~\cite{BahooBaBoDuMu17,bkls-cvpufv-14},
planar convex hulls~\cite{bklss-sttosba-14,de-otst2dchp-14},
general plane-sweep algorithms~\cite{ElmasryKa16}, or triangulating
simple polygons~\cite{abbkmrs-mcasp-11,ak-tstanlnp-13,akprr-tstotsp-17}.
We refer the reader to~\cite{k-mca-15} for an overview of different
ways of modeling computation in the presence of space constraints.

Let us specify our model more precisely: we are given a set $P$ of
$n$ \emph{point sites} in the plane. The set $P$ is stored in a
read-only array that allows random access. Furthermore, we may use
$O(s)$ words of memory (for a parameter $s \in \{1, \ldots, n\}$) for
reading and writing. We assume that all the data items and pointers
are represented by $\Theta(\log n)$ bits. Other than this, the model
allows the usual word RAM operations. 

We consider the problem of computing various Voronoi diagrams for
$P$, namely the \emph{nearest site Voronoi diagram} $\NVD(P)$, the
\emph{farthest site Voronoi diagram} $\FVD(P)$, and the family of all
\emph{higher-order Voronoi diagrams} up to a given order
$K\in \{1, \dots,  O(\sqrt{s})\}$.  For most values of $s$, the
output cannot be stored explicitly. Thus, we require that the
algorithm reports the edges of the Voronoi diagrams one by one in a
write-only data structure, separately for each diagram, in increasing
order of $k$. Once written, the output cannot be read or further
modified.  Note that we may report edges of each Voronoi diagram
in any order, but we are not allowed to report an edge more than once.
\vspace{-0.1cm}
\paragraph{Previous Work and Our Results.}
If we forego memory constraints, it is well known that both $\NVD(P)$
and $\FVD(P)$ can be computed in $O(n \log n)$ time using $O(n)$ 
space~\cite{AurenhammerKlLe13,bcko-cgaa-08}. For computing a single 
Voronoi diagram 
of order $k$, the best known randomized algorithm takes
$O\left(n\log n + nk \, 2^{O(\log^*k)}\right)$ time and $O(nk)$
space~\cite{Ramos99}, while the best known 
deterministic algorithm takes $O(n\log n + nk\log k)$ time and
$O(nk)$ space~\cite{chan2000random,ChanTs16}.\footnote{This algorithm
uses the rather involved dynamic planar convex hull structure of 
Brodal and Jacob~\cite{BrodalJa02}. If the reader prefers a more 
elementary method, 
we can substitute the slightly slower, but much simpler, previous 
result by the same authors. The running time then becomes
$O(n\log n + nk\log k\log\log k)$~\cite{BrodalJa00,ChanTs16}.}
For any given $K \in \{1, \dots, n-1\}$, the family
of all higher-order Voronoi diagrams of order $k= 1, \dots, K$ can be
computed in $O(nK^2+n\log{n})$ deterministic time using $O(K^2(n-K))$ 
space~\cite{AggarwalGuSaSh89, Lee82}.
 
In the literature, there are very few memory-constrained algorithms 
that compute Voronoi diagrams.  Asano~\etal~\cite{amrw-cwagp-10}
showed that $\NVD(P)$ can be found in $O(n^2)$ time using $O(1)$
words of workspace. Korman~\etal~\cite{KormanMuReRoSeSt15} gave a
time-space trade-off for computing $\NVD(P)$.  Their algorithm is
based on random sampling and achieves an expected running time of
$O((n^2/s)\log s + n\log s \log^* s)$ using $O(s)$ words of workspace.
We provide time-space trade-offs that improve and generalize the
known memory-constrained algorithms for computing Voronoi diagrams.
We believe that our method is simpler and more flexible than previous methods.
In Section~\ref{sec:o1_nvd_fvd}, we show that the approach of
Asano~\etal~\cite{amrw-cwagp-10} can be used to compute $\FVD(P)$. In
Section~\ref{sec_trade}, we introduce a new time-space trade-off for
computing $\NVD(P)$ and $\FVD(P)$. Unlike the result of
Korman~\etal~\cite{KormanMuReRoSeSt15}, this new algorithm is
deterministic and slightly faster. It runs in $O((n^2/s)\log s)$ time
using $O(s)$ words of workspace, thus saving a $\log^* s$ factor for
large values of $s$. 

Finally, in Section~\ref{sec_high}, we use the $s$-workspace
algorithm from Section~\ref{sec_trade} as a building block in a new
pipelined algorithm.  For any given $K \in O(\sqrt{s})$, this
algorithm computes the family of all higher-order Voronoi diagrams of order
$k = 1, \dots, K$ in total expected time 
$O\bigl(\frac{n^2 K^5}{s}(\log s + K \, 2^{O(\log^* K)})\bigr)$
or in total deterministic time 
$O\bigl(\frac{n^2 K^5}{s}(\log s + K \log K)\bigr)$,
using $O(s)$ words of workspace.
To compute the edges of a Voronoi diagram of order $k$, we use the
edges of the diagram of order $k-1$. However, this needs to be
coordinated carefully, to prevent edges from being reported
multiple times and to not exceed the space budget.

\section{Preliminaries and Notation}

Throughout the paper we denote by $P = \{p_1, \dots, p_n\}$ a
set of $n\geq 3$ \emph{sites} in the plane. We assume general position, 
meaning that no three sites of $P$ lie on
a common line and no four sites of $P$ lie on a common
circle. To fix our terminology, we recall some classic and well-known
properties of Voronoi diagrams~\cite{AurenhammerKlLe13,bcko-cgaa-08}. 

\begin{figure}[tb]
	\centering
	\subcaptionbox{$\NVD(P) = \VD^1(P)$ \label{fig:1a}}{\includegraphics[page=1]{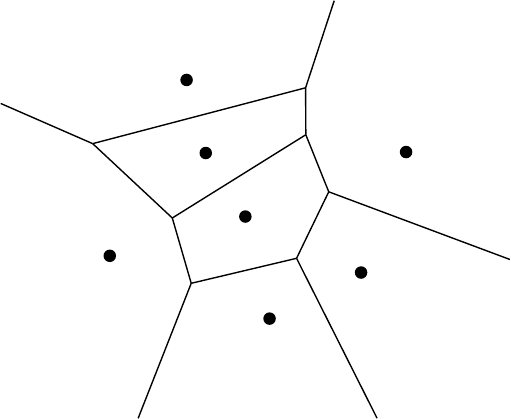}}
	\subcaptionbox{$\FVD(P) = \VD^{n-1}(P)$ \label{fig:1b}}{\includegraphics[page=2]{1}}
	\subcaptionbox{$\VD^2(P)$ \label{fig:1c}}{\includegraphics[page=3]{1}}
	\caption{For $P$, a set of planar sites (a) The nearest site Voronoi diagram (b) The farthest site Voronoi diagram (c) The Voronoi diagram of order $2$.}
\end{figure}

The \emph{nearest site Voronoi diagram}
for $P$, $\NVD(P)$, is obtained by classifying the points
in the plane according to their nearest neighbor in
$P$. For each site 
$p \in P$, the open set
of points in $\R^2$ with $p$ as their unique nearest
site in $P$ is called the \emph{Voronoi cell} of $p$.
For any two sites $p, q \in P$, the \emph{bisector} $B(p,q)$ of $p$
and $q$ is defined as the line containing all points in the plane 
that are equidistant to $p$ and $q$.
The \emph{Voronoi edge} for $p$, $q$ consists of all points 
in the plane with $p$ and $q$ as their only two nearest sites. 
If it exists, the Voronoi edge for $p$ and $q$ is a subset of the 
bisector $B(p,q)$ of $p$ and $q$.
Our general position assumption, and the fact that $n \geq 3$, guarantee 
that each Voronoi edge is an open line segment or a halfline. 
\emph{Voronoi vertices} are the points in the plane that have exactly three nearest sites
in $P$. Again by general position, we have that every point in $\R^2$ is either a
Voronoi vertex, or lies on a Voronoi edge or in a Voronoi cell.
The Voronoi vertices and the Voronoi edges form the set of vertices
and edges of a plane graph whose faces are the Voronoi cells. 
This graph is called the nearest site Voronoi diagram for $P$,
$\NVD(P)$; see Figure~\ref{fig:1a}.
It has $O(n)$ vertices, $O(n)$ edges, and $n$ cells.

The \emph{farthest site Voronoi diagram} for $P$, $\FVD(P)$,
is defined analogously. Farthest Voronoi cells, edges, and vertices
are obtained by replacing the term ``nearest site'' by
the term ``farthest site'' in the respective definitions.
Again, the farthest Voronoi vertices and edges constitute the
vertices and edges of a plane graph, called $\FVD(P)$.
As before, it has $O(n)$ vertices and $O(n)$ edges. However,
unlike in $\NVD(P)$, in $\FVD(P)$ it is not necessarily the
case that all sites in $P$ have a corresponding cell in 
$\FVD(P)$. Indeed, the sites with non-empty farthest Voronoi
cells are exactly the sites on the \emph{convex hull}
of $P$, $\conv(P)$. Furthermore, all cells in $\FVD(P)$ are
unbounded. Hence, $\FVD(P)$, considered as a plane graph,
is a tree; see Figure~\ref{fig:1b}.

Now, for $k \in \{1, \dots, n-1\}$, the Voronoi diagram \emph{of order $k$}
for $P$ is obtained by classifying the points in the plane into 
cells, edges, and vertices according to the \emph{set} of
sites in $P$ that achieve the $k$ smallest distances.
We denote the Voronoi diagram of order $k$ for $P$ by $\VD^k(P)$; see 
Figure~\ref{fig:1c}.
Observe that $\NVD(P) = \VD^1(P)$ and $\FVD(P) = \VD^{n-1}(P)$. 
For each set $Q \subset P$ of $k$ sites from
$P$, we denote the \emph{Voronoi cell of order $k$} for $Q$ by $C^k(Q)$.
It is known that $\VD^k(P)$ is a plane graph of
complexity $O(k(n-k))$~\cite{AurenhammerKlLe13,Lee82}. For simplicity,
the cell of $p\in P$ in $\NVD(P)$ and $\FVD(P)$ are denoted,
respectively, by $C^1(p)$ and $C^{n-1}(p)$.
We will denote the boundary of a cell $C$ by $\partial C$.
We will give more properties of higher-order Voronoi diagrams
in Section~\ref{sec_high}.

\section{A Constant Workspace Algorithm for FVDs and NVDs}
\label{sec:o1_nvd_fvd}

\begin{figure}
  \centering
    \includegraphics{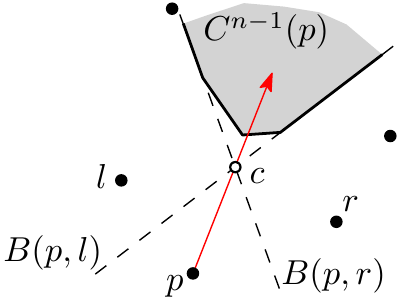}
\caption{An illustration of Facts~\ref{fa:relation of CH and FVD} 
  and~\ref{fa:ray between two unbounded edges}: The sites
$l,p,r \in P$ are consecutive on $\conv(P)$. The boundary 
$\partial C^{n-1}(p)$ contains a subset of $B(p,l)$ and of $B(p,r)$. 
The ray from $p$ toward $c = B(p,l) \cap B(p,r)$ intersects 
$\partial C^{n-1}(p)$.}
\label{fig:2}
\end{figure}

We are given a set 
$P=\{p_1, \dots, p_n\}$ of $n$ sites in
the plane stored in a 
read-only array to which we have random access. Our task is to 
report the edges of $\NVD(P)$ and of 
$\FVD(P)$ using only a constant amount of 
additional workspace. First, we show how 
to find a single edge of a given cell of $\NVD(P)$ 
or of $\FVD(P)$. Then, we repeatedly use this procedure
to find all the edges of $\NVD(P)$ and $\FVD(P)$. 
We summarize the properties of $\FVD(P)$ that are 
relevant to our algorithms in the following
two facts. More details can be found, e.g., in the book 
by Aurenhammer, Klein, and Lee~\cite{AurenhammerKlLe13}.
See \textup{Figure~\ref{fig:2}} for an illustration.

\begin{fact}\label{fa:relation of CH and FVD}
Let $P$ be a set of $n$ point sites in the plane in general position,
and let $p \in P$. The cell $C^{n-1}(p)$ is not empty if and only if 
$p$ lies on the convex hull of $P$. In this case,
the farthest Voronoi cell of $p$ is unbounded. Furthermore,  
if $r, l\in P$ are the
two adjacent sites of $p$ on $\conv(P)$, then $C^{n-1}(p)$ contains
an unbounded edge for $p$ and $l$ and an unbounded edge
for $p$ and $r$. These edges are subsets
of $B(p, l)$ and of $B(p,r)$, respectively.
\end{fact}

\begin{fact}\label{fa:ray between two unbounded edges}
Let $P$ be a set of $n$ point sites in the plane in general position.
Let $l, p, r \in P$ be three consecutive sites on $\conv(P)$,
and let $c$ be the intersection of $B(p,l)$ and $B(p,r)$. 
Then, the ray from 
$p$ toward $c$ intersects $\partial C^{n-1}(p)$ \textup{(}not necessarily at $c$\textup{)}. 
\end{fact}

\begin{lemma}\label{le:a ray intersecting the boundary}
Let $P$ be a set of $n$ point sites in the plane in general position.
Suppose that $P$ is given in a read-only array.
For any $p \in P$, in $O(n)$ time and
using constant workspace, we can
determine whether $C^{n-1}(p)$ is not empty.
If so, we can also find a ray that intersects 
$\partial C^{n-1}(p)$.
\end{lemma}

\begin{proof}
By Fact~\ref{fa:relation of CH and FVD}, it suffices
to check whether $p$ lies inside $\conv(P)$. This
can be done using simple \emph{gift-wrapping}: 
pick an arbitrary site $q \in P\setminus\{p\}$. Scan 
through $P$ and find the sites $p_\text{cw}$ and 
$p_\text{ccw}$ in $P$ which make, respectively, the 
largest clockwise angle and the largest 
counterclockwise angle with the ray $pq$, such that both 
angles are at most $\pi$. 
Both $p_\text{cw}$ and $p_\text{ccw}$ are easily obtained in 
$O(n)$ time using constant workspace. If the 
cone $p_\text{cw}pp_\text{ccw}$ that contains $q$ has an 
opening angle larger than $\pi$, then $p$ is inside $\conv(P)$ 
and consequently $C^{n-1}(p)$ is empty. 
Otherwise, $p$ is on $\conv(P)$, with $p_\text{cw}$ 
and $p_\text{ccw}$ as its two neighbors. 
By Fact~\ref{fa:ray between two unbounded edges}, the ray
from $p$ through $B(p,  p_\text{cw}) \cap B(p, p_\text{ccw})$ 
intersects $\partial C^{n-1}(p)$.
\end{proof}

\begin{lemma}\label{le:finding an edge}
Let $P$ be a planar $n$-point set in general position 
in a read-only array. Suppose we are given a site $p \in P$ and a ray $\gamma$ that
emanates from $p$ and intersects 
$\partial C^1(p)$. Then, we can 
report an edge $e$ of $C^1(p)$ 
 that intersects $\gamma$, 
in $O(n)$ time using $O(1)$ words of workspace.
An analogous statement holds for $\FVD(P)$.
\end{lemma}

\begin{proof}
Among all bisectors $B(p, p')$, for $p' \in P \setminus \{p\}$,
we find a bisector $B^* = B(p, p^*)$ that intersects $\gamma$ closest to 
$p$.\footnote{If $\gamma$ happens to intersect a vertex of $C^1(p)$,
there are two such bisectors. Otherwise, $B^*$ is unique.} 
We can find $B^*$ by scanning the sites of 
$P$ and maintaining a closest bisector in each step.
The edge $e$ is a subset of~$B^*$.
To find the portion of $B^*$ that forms a Voronoi edge in $\NVD(P)$, we do a second scan of 
$P$. For each $p' \in P \setminus \{p, p^*\}$, we check where $B(p, p')$ 
intersects $B^*$. Each such intersection cuts a
piece from $B^*$ that cannot appear in $\NVD(P)$, namely
the part of $B^*$ that is closer to $p'$ than to $p$. 
After scanning all the sites of $P$, the remaining portion of $B^*$ is exactly $e$.
Since the current piece of $B^*$ in each step is connected,
we need to store only at most two endpoints in each step.
Overall, we can find the edge $e$ 
of $C^1(p)$ that intersects $\gamma$ in $O(n)$ time using 
$O(1)$ words of workspace.

The procedure for $\FVD(P)$ is analogous, but we take
$B^*$ to be the bisector intersecting $\gamma$ farthest from $p$,
and we cut from $B^*$ the pieces that are closer to $p$ than to any other site.
\end{proof}

\begin{theorem}\label{T:constant workspace algo}
Suppose we are given a planar $n$-point set 
$P=\{p_1, \dots, p_n\}$ in general position in
a read-only array. We can find all the edges of 
$\NVD(P)$ in $O(n^2)$ time 
using $O(1)$ words of workspace.
The same holds for $\FVD(P)$.
\end{theorem}

\begin{proof}
First, we restate the strategy for 
$\NVD(P)$ that was proposed by Asano~\etal~\cite{amrw-cwagp-10},
and then we show how to adapt it for $\FVD(P)$.

We go through the sites in $P$.
In step $i$, we process $p_i \in P$ to detect all 
edges of $C^1(p_i)$. For this, we need a ray $\gamma$ to 
apply Lemma~\ref{le:finding an edge}.  
We choose $\gamma$ as the ray from $p_i$ 
to an arbitrary site of $P \setminus \{p_i\}$. 
This ensures that $\gamma$ intersects 
$\partial C^1(p_i)$.  Now, we use 
Lemma~\ref{le:finding an edge} to find an edge $e$ 
of $C^1(p_i)$ that intersects $\gamma$.
We consider the ray $\gamma'$ from $p_i$ through the left endpoint 
of $e$ (if it exists), and we apply Lemma~\ref{le:finding an edge} 
to find the adjacent edge $e'$ of $e$ in $C^1(p_i)$.\footnote{Note 
that the bisector that defines the left endpoint of $e$ is also the 
bisector that is spanned by $e'$. Thus, the first scan 
of the input in Lemma~\ref{le:finding an edge}, for finding 
the line spanned by $e'$, is not strictly necessary. However, 
since we must scan the input anyway to determine the endpoint of $e'$,
we chose to present the algorithm as doing two scans.
This keeps the presentation more uniform, at the expense of only 
a constant factor in the running time. The same comment also applies
to our later algorithms.}
The ray $\gamma'$ hits both $e$ and $e'$, so we perform a symbolic 
perturbation to $\gamma'$ so that only $e'$ is hit.
We repeat this procedure to find further
edges of $C^1(p_i)$, in counterclockwise direction.
This continues until we return to $e$ or 
until we find an unbounded edge of $C^1(p_i)$. 
In the latter case, we start again from the right 
endpoint of $e$ (if it exists), and we find the remaining
edges of $C^1(p_i)$ in clockwise direction. 

Since each edge of $\NVD(P)$ is incident to two Voronoi cells, 
this process will detect each edge twice. 
To avoid repetitions, whenever we find an edge $e$ of $C^1(p_i)$ 
with $e \subseteq B(p_i, p_j)$, we 
report $e$ if and only if $i<j$. 
Since $\NVD(P)$ has $O(n)$ edges, and 
reporting one edge takes $O(n)$ time and $O(1)$ words 
of workspace, the result follows.

For $\FVD(P)$, the procedure is almost the same. However, when going
through the sites in $P$, for each $p_i \in P$, we first check 
if $C^{n-1}(p_i)$ is non-empty, using 
Lemma~\ref{le:a ray intersecting the boundary}. If so, the algorithm
from the lemma also gives us a ray $\gamma$ that intersects 
$\partial C^{n-1}(p_i)$. From here, we proceed exactly as for $\NVD(P)$
to find the remaining edges of $C^{n-1}(p_i)$.
\end{proof}

\section{Obtaining a Time-Space Trade-off}\label{sec_trade}

Now we adapt the previous algorithm to a time-space trade-off.
Suppose we have $O(s)$ words of workspace at our 
disposal, for some $s \in \{1, \dots, n\}$.\footnote{The assumption 
that we have $O(s)$ words instead of exactly $s$ words of workspace
is mostly for the sake of a simple presentation. Thus, when 
describing our algorithm, we can ignore constant factors in the space
usage. The precise constant is a function that only depends on the implementation of the algorithm.}
As before, we are given a planar $n$-point set $P = \{p_1, \dots, p_n\}$ 
in general position in a read-only array, and we would like to report
all edges of $\NVD(P)$ or $\FVD(P)$ as quickly as possible.
While the algorithm from Section~\ref{sec:o1_nvd_fvd} needs two
passes over the input to find a single edge of the Voronoi diagram,
the idea now is to exploit the additional workspace in order to find $s$ edges 
of the Voronoi diagram in parallel using two passes.
For this, we first show how to find simultaneously a single edge 
for $s$ different cells of $\NVD(P)$ or of $\FVD(P)$.

\begin{lemma}\label{le:finding s edges}
Suppose we are given a set $V=\{v_1, \dots, v_s\}$
of $s$ sites in $P$, and for each $i = 1, \dots, s$, a ray 
$\gamma_i$ emanating from $v_i$ such that $\gamma_i$ intersects the 
boundary of $C^1(v_i)$.  Then, we can 
report for each $i = 1, \dots, s$, an edge $e_i$ of $C^1(v_i)$  
that intersects $\gamma_i$, in $O(n\log s)$ total time using $O(s)$ 
words of workspace.
An analogous statement holds for $\FVD(P)$.
\end{lemma}

\begin{proof}
The algorithm has two phases. In the first phase, for 
$i = 1, \dots, s$, we find the bisector $B^*_i$ that contains $e_i$, and 
in the second phase, for $i = 1, \dots, s$, we find $e_i$, i.e., the 
portion of $B^*_i$ that is in $\NVD(P)$.

The first phase proceeds as follows:
we group $P$ into \emph{batches} $Q_1, Q_2, \dots, Q_{n/s}$ of $s$
consecutive sites (according to the order in the input array).
First, we compute $\NVD(V \cup Q_1)$. Since $|V \cup Q_1|  \leq 2s$, 
this takes $O(s\log s)$ time using $O(s)$ words of workspace. 
Now, for $i = 1, \dots, s$,
we find the edge $e_i'$ of $\NVD(V \cup Q_1)$ that intersects 
$\gamma_i$ closest to $v_i$, and we 
store the bisector $B_i'$ that contains $e_i'$. 
This can be done in total time $O(|V \cup Q_1|)$, since each ray
originates in a unique Voronoi cell and since we can simply
traverse the whole diagram $\NVD(V \cup Q_1)$ to find the 
intersection points.
Then, for $j = 2, \dots, n/s$,
we again compute 
$\NVD(V \cup Q_j)$. For $i = 1, \dots, s$, we find the edge 
in $\NVD(V \cup Q_j)$ that intersects $\gamma_i$ closest to $v_i$,
in total time $O(|V \cup Q_j|)$. We update
$B_i'$ to the bisector that contains this edge if and only if its intersection
with $\gamma_i$ is closer to $v_i$ than for the current $B_i'$. 
We claim that after all batches $Q_1, \dots, Q_{n/s}$ have been scanned, 
$B_i'$ is the desired bisector $B_i^*$.
To see this, let $B_i^* = B(v_i, p)$, for a site
$p \in P \setminus\{v_i\}$. Then, for the batch
$Q_j$ with $p \in Q_j$, the Voronoi diagram $\NVD(V \cup Q_j)$
contains an edge on $B_i^*$. Furthermore, by definition,
no other bisector intersects $\gamma_i$ closer to $v_i$ than $B_i^*$.

In the second phase, we again group $P$ into batches
$Q_1, \dots, Q_{n/s}$ of size $s$. 
We again compute $\NVD(V \cup Q_1)$. 
For $i = 1, \dots, s$, we find the portion 
of $B_i^*$ inside the cell of $v_i$ in $\NVD(V \cup Q_1)$, 
and we store it in $e_i$. Then, for $j= 2, \dots, n/s$, 
we compute $\NVD(V \cup Q_j)$, and for $i = 1, \dots, s$, 
we update the endpoints of $e_i$ to the 
intersection of the current $e_i$ and the cell of $v_i$ 
in $\NVD(V \cup Q_j)$. After processing $Q_j$, 
there is no site in $V \cup \bigcup_{m=1}^{j} Q_m$  that is closer to $e_i$ than $v_i$.
Thus, at the end of the second phase, $e_i$ 
is the edge of $C^1(v_i)$ that intersects $\gamma_i$. 
Due to the properties of the Voronoi diagram, throughout the 
algorithm, $e_i$ is a connected subset of
$B_i^*$ (i.e., a ray or a line segment), and it can be described 
with $O(1)$ words of workspace. 

In total, we construct $O(n/s)$ Voronoi diagrams, each with at most
$2s$ sites. Since we have $O(s)$ words of workspace available, it takes
$O(s \log s)$ time to compute a single 
Voronoi diagram. 
Thus, the total running time is $O(n\log s)$. At each point in time, we 
have $O(s)$ sites in workspace and a constant amount of information for each site, including the Voronoi diagram of these sites,
so the space bound is not exceeded.
The proof for $\FVD(P)$ is analogous.
\end{proof}

Now we describe our time-space trade-off algorithm. At each point in time,
we have a set $V$ of $s$ sites in workspace. We use 
Lemma~\ref{le:finding s edges} to produce a new edge for each site
in $V$. Once all edges for a site $v \in V$ have been found, we discard
$v$ from $V$ and replace it with a new site from $P$ (we say that $v$ has been processed completely).
We stop this process as soon as all but fewer than $s$ sites have been processed completely. At this point,
we do not use Lemma~\ref{le:finding s edges} any longer. This is because 
Lemma~\ref{le:finding s edges} needs two passes of the input to find 
a single new edge for each site in $V$. Thus, if there is a cell with
many edges, 
too many passes will be necessary.
To avoid this, we will need a different method for finding the edges
of the remaining cells, see below.
We call these remaining cells \emph{big}, and the other
cells  \emph{small}.
By definition, all small cells have $O(n/s)$ edges, 
but big cells may have a lot more edges (even though this does not have
to be the case). 

In order to avoid doubly reporting edges, our algorithm is split into 
three \emph{phases}. In the first phase, we 
process the whole input to identify the big cells (no edge is reported 
in this phase). The second phase scans the input again 
and reports all edges incident to at least one small cell.
The third phase reports edges incident to 
two big cells.
\paragraph{First phase.}
The aim of this phase is to find the big cells. We 
describe how we use Lemma~\ref{le:finding s edges} in more detail. 
We scan all sites with non-empty Voronoi cells. For $\NVD(P)$,
since all sites have a non-empty cell, 
we can scan them sequentially. The starting ray is constructed in 
the same way as in Theorem~\ref{T:constant workspace algo}. 
For $\FVD(P)$, by Fact~\ref{fa:ray between two unbounded edges}, we 
need to find the sites on the convex hull of $P$.  For this, we use
the algorithm of Darwish and Elmasry~\cite{de-otst2dchp-14} that 
reports the sites on the convex hull of $P$ in clockwise order 
in $O(\frac{n^2}{s\log n} + n\log s)$ time using $O(s)$ words of workspace.  
We run the Darwish-Elmasry algorithm until $s$ sites on the convex hull 
have been identified. Then, we suspend the convex 
hull computation and process those sites. Whenever more sites are 
needed, we simply resume the convex hull algorithm. 
Since the convex hull is reported in clockwise order, we know the 
two neighbors for each site on the convex hull and we can find 
a starting ray using Fact~\ref{fa:ray between two unbounded edges} 

At each point in time, our Voronoi algorithm has
$s$ sites from $P$ with non-empty cells in memory. 
We apply Lemma~\ref{le:finding s edges} to 
compute one edge on the cell of each such site. After that, we 
iteratively update the rays of all sites in memory to find the next 
edge of each cell, as in Theorem~\ref{T:constant workspace algo}. 
Whenever all edges of a cell have been found, we remove the 
corresponding site from memory, and we replace it with the next 
relevant site; see Figure~\ref{fig:3}. Since $(1)$ the Voronoi diagram of $P$
has $O(n)$ edges, $(2)$ in each iteration we produce $s$ edges,
and $(3)$ each edge is produced at most twice, it follows that after $O(n/s)$ iterations, 
fewer than $s$ sites remain in memory. All other sites of $P$ must have 
been processed. 

Thus, after the first phase, we have identified all 
big cells (those that have not been processed fully). Since there are
at most $s$ of them, we can store the corresponding sites explicitly 
in a table $\mathcal{B}$. We sort those sites according to their indices,
so that membership in $\mathcal{B}$ can be tested in $O(\log s)$ time. 

\begin{figure}
  \centering
    \includegraphics{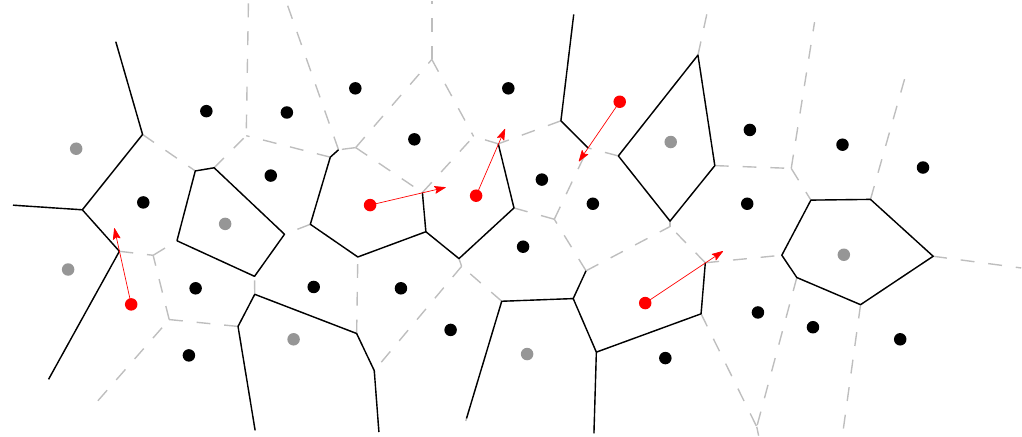}
\caption{The state of the algorithm at the end of iteration 9 of 
  applying Lemma~\ref{le:finding s edges}, for a set $P$ of $35$ 
  sites and workspace of size $O(s)= O(\lfloor\log n\rfloor)$. 
  The black segments are the edges of $\NVD(P)$ that have already 
  been found. The gray and the red sites represent, respectively, 
  the sites which have been fully processed and those which are currently 
  in the workspace.}
\label{fig:3}
\end{figure}
\paragraph{Second phase.}
The second phase is very similar to the first 
one.\footnote{Indeed, these two phases could be merged into one. 
However, as we will see below, it is not straightforward to do so  for higher-order 
Voronoi diagrams. Thus, for consistency, we split the 
two phases even for $k = 1$ and $k = n-1$.} Pick $s$ sites to process; 
repeatedly use Lemma~\ref{le:finding s edges} to find edges 
for each site; 
once all edges of a site $v$ have been found, replace $v$ with the
next site; continue until only big cells remain. The main difference 
now is  we report some Voronoi edges (making sure 
that every edge is reported exactly once). More precisely,  
suppose that we discover a Voronoi edge $e$ while 
scanning the cell $C_i$ of a site $v_i$, and that $e$ is also 
incident to the cell $C_j$ of the site $v_j$. 
Then, we report $e$ only if 
one of the following conditions holds:
\setlist{nolistsep}
\begin{enumerate}\renewcommand{\labelenumi}{(\roman{enumi})}
\item both $C_i$ and $C_j$ are small and $i < j$; or 
\item $C_i$ is small and $C_j$ is big.
\end{enumerate}
\paragraph{Third phase.}
The purpose of the third phase is to report every 
Voronoi edge that is incident to two big cells. 
For this, we compute the Voronoi diagram of 
the sites of big cells, in $O(s \log s)$ time.
Let $E_\mathcal{B}$ denote the set of its edges. The edges of $E_\mathcal{B}$ that are 
also present in the Voronoi diagram of $P$ need to be reported 
(the edges may need to be truncated).

In order to determine which edges of $E_\mathcal{B}$ remain in the diagram, 
we proceed similarly as in the second scan of 
Lemma~\ref{le:finding s edges}: in each step, we compute the 
Voronoi diagram $\mathcal{V}$ of $\mathcal{B}$ and a batch of $s$ sites from 
$P$. For each edge $e$ of $E_\mathcal{B}$, we check whether $e$ is 
cut off in $\mathcal{V}$. If so, we 
update the endpoints of $e$ to the intersection of $e$ and 
the cell for one of the sites defining $e$. After all edges have been checked, 
we continue with the next batch of $s$ sites from $P$. After 
processing all the sites of $P$, the remaining $O(s)$ edges 
in $E_\mathcal{B}$ that have not become empty constitute all 
the edges of the Voronoi diagram of $P$
that are incident to two big cells.
In contrast to Lemma~\ref{le:finding s edges}, 
we report $O(s)$ edges that are not necessarily incident to $s$ 
different cells. 

\begin{theorem}\label{T:nvd_fvd_ts}
Let $P = \{p_1, \dots, p_n\}$ be a 
planar $n$-point set in general position 
stored in a read-only array.
Let $s$ be a parameter in $\{ 1, \dots, n \}$.
We can report 
all  edges of $\NVD(P)$ in 
$O((n^2/s) \log s)$ time using $O(s)$ words 
of workspace.
An analogous result holds for $\FVD(P)$. 
\end{theorem}
\begin{proof}
Lemma~\ref{le:finding s edges} guarantees that the edges reported 
in the second phase are part of $\NVD(P)$. Also, conditions (i) and 
(ii) ensure that no edge is reported twice. 
Clearly, if an 
edge $e \in \NVD(P)$ is incident to two big cells, the same edge 
(possibly a superset) must be present in $\NVD(\mathcal{B})$. 
For the reverse inclusion, first note that since $\mathcal{B} \subset P$,
an edge incident to two big cells that is not present in $\NVD(\mathcal{B})$ 
cannot be present in $\NVD(P)$. Furthermore, for each edge $e$ of $\NVD(\mathcal{B})$, 
we consider all sites of $P$ and we remove only the 
portions of $e$ that cannot be present in $\NVD(P)$.

Finally, we need to analyze the running time. 
The most expensive part of the algorithm lies in the $O(n/s)$ 
invocations of Lemma~\ref{le:finding s edges} during the first and 
the second phase. Other than that, creating the table $\mathcal{B}$ needs 
$O(s\log s)$ time, and we perform $O(n)$ lookups in $\mathcal{B}$, two for 
each edge of $\NVD(P)$. Each lookup needs $O(\log s)$ time, 
so $O(n\log s)$ time in total. The third phase does a single scan over 
the input, and it computes a Voronoi diagram for each batch of 
$s$ sites, which totally takes $O(n\log s)$ 
time. Thus, the running time of the algorithm
is $O((n^2/s)\log s)$.

At each point during the algorithm, we store 
only $s$ sites that 
are currently being processed (along with a constant amount of  
information attached to each such site), the table $\mathcal{B}$ of at most 
$s$ sites, the batch of $s$ sites being processed (and the 
associated Voronoi diagram). 
All of this can be stored using $O(s)$ words of workspace, as 
claimed. 

For $\FVD(P)$, the approach is analogous. The only difference is 
that now we must also find the convex hull of
$P$. With the algorithm of Darwish and Elmasry~\cite{de-otst2dchp-14},
this takes $O((n^2/s) \log s)$ time for $O(s)$ words of workspace,  
so the asymptotic running time does not increase.
\end{proof}

\section{Higher-Order Voronoi Diagrams}\label{sec_high}

\begin{figure}[tb]
  \centering
  \subcaptionbox{\label{fig:4a}}{\includegraphics[page=1]{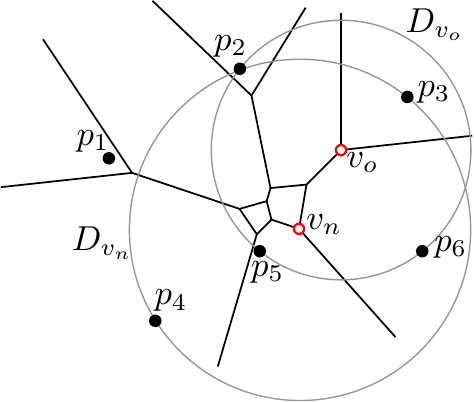}}\hspace{0.5in}
  \subcaptionbox{\label{fig:4b}}{\includegraphics[page=2]{4}}
\caption{The diagram $\VD^k(P)$ for $k =3 $ and $P =\{p_1,\dots, p_6\}$. 
  (a) The interior of the disk $D_{v_n}$ with center $v_n$ 
  contains $k-1$ sites $\{p_5, p_6\}$, so the $k$-vertex $v_n$ is 
  new. The interior of the disk $D_{v_o}$ with center $v_o$ contains 
  $k-2$ sites $\{p_3\}$, so the $k$-vertex $v_o$ is old. 
  (b) The $k$-cell $C_{4,5,6}$ is the cell of $\{p_4, p_5, p_6\}$. 
  The $k$-edge $e_1$  is represented by the set $\{p_5,p_6\}$
  containing the $k-1$ sites closest to $e_1$, the two sites $p_3$ 
  and $p_4$ that are equidistant to $e_1$, and the site $p_2$ 
  that defines the $k$-vertex $v_n$. Since $v_n$ is a new $k$-vertex, 
  the site $p_2$ is not among the $k-1$ closest sites to $e_1$. 
  The $k$-edge $e_2$ of the $k$-cell $C_{2,3,6}$ for $\{p_2,p_3,p_6\}$ 
  is represented by the set $\{p_2,p_3\}$ of $k-1$ sites closest to $e_2$, 
  the two sites $p_5$ and $p_6$ that are equidistant to $e_2$, and the 
  site $p_2$ that defines the $k$-vertex $v_o$. Since $v_o$ is an old $k$-vertex, 
  the site $p_2$ is among the $k-1$ closest sites to $e_2$.}
\end{figure}

We now consider computing higher-order Voronoi diagrams~\cite{Lee82}. 
More precisely, we are given an integer $K \in O(\sqrt{s})$,
and we would like to report the family of all higher-order Voronoi 
diagrams of order $k=1, \dots, K$, 
where we have $O(s)$ words of workspace at our 
disposal, for some $s \in \{1, \dots, n\}$.
For this, we generalize our approach from the previous 
section, and we combine it with a recursive procedure:
for $k = 1, \dots, K-1$, we compute the edges of $\VD^{k+1}(P)$ by using 
previously computed edges of $\VD^{k}(P)$. To make efficient
use of the available memory, we perform the computation of 
the diagrams $\VD^1(P), \VD^2(P), \dots, \VD^K(P)$ in a pipelined 
fashion, so that in each stage, the necessary edges of the previous 
Voronoi diagrams are at our disposal
and the total memory usage remains $O(s)$.

We begin with some more background on higher-order Voronoi diagrams.
Let $x \in \R^2$ be a point in the plane.
The \emph{distance order} for $x$ is the sequence of sites in $P$
ordered according to their distance from $x$, from closest to 
farthest. By our general position assumption, there are at most 
three sites in $P$ with the same distance to $x$.
We call a cell $C$ of $\VD^k(P)$ a \emph{$k$-cell}, and
we represent it as the set of $k$ sites that are closest
to all points in $C$. Similarly, we call a vertex $v$
of $\VD^k(P)$ a \emph{$k$-vertex}. It is known that there exists 
a disk $D_v$ with center $v$ such that 
$|\partial D_v \cap P| = 3$ and 
$|\mathring{D}_v \cap P| \in \{k-2, k-1\}$, where 
$\partial D_v$ is the boundary and $\mathring{D}_v$ is the interior 
of $D_v$. We call $v$ an \emph{old} vertex if 
$|\mathring{D}_v \cap P| = k-2$, and a \emph{new} vertex if
$|\mathring{D}_v \cap P| = k-1$; see Figure~\ref{fig:4a}.
We represent $v$ by the set $D_v \cap P$, marking the 
sites on $\partial D_v$.
Finally, the edges of $\VD^k(P)$
are called \emph{$k$-edges}. We represent them in a
somewhat unusual manner: each edge 
of $\VD^k(P)$ is split into two directed \emph{half-edges},
such that the half-edges are oriented in opposing directions 
and such that each half-edge is \emph{associated} with the $k$-cell to its
left. A half-edge $e$ is represented by $k+3$ sites of $P$:
the $k-1$ sites closest to $e$, the two sites that come next
in the distance order for the points on $e$ and are equidistant to 
$e$, and one more 
site for each endpoint of $e$, to define the corresponding
$k$-vertices.
For each endpoint $v$ of $e$, there are two cases: if $v$ is an old 
vertex, the third site defining $v$ is among the $k-1$ sites closest 
to $e$, and if $v$ is a new vertex, the third site is not among those $k-1$ sites; see Figure~\ref{fig:4b}.
The order of the endpoints encodes the direction of the half-edge.
The half-edge is directed from the \emph{tail} vertex to
the \emph{head} vertex.

\begin{figure}[tb]
  \centering
  \subcaptionbox{\label{fig:5a}}{\includegraphics[page=1]{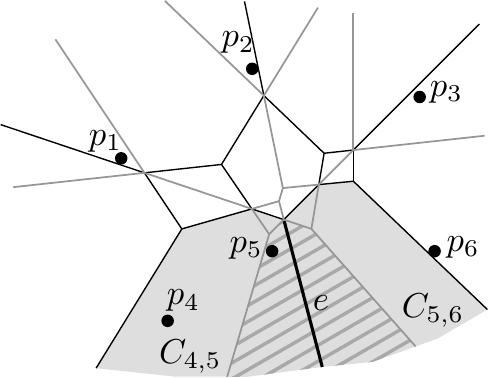}}\hspace{0.5in}
  \subcaptionbox{\label{fig:5b}}{\includegraphics[page=2]{5}}
\caption{The diagrams $\VD^k(P)$ (black) and $\VD^{k+1}(P)$ (gray), 
  for $k = 2$ and $P = \{p_1,\dots, p_6\}$. (a) The $k$-cells 
  $C_{4,5}=C^k(\{p_4, p_5\})$ and $C_{5,6}=C^k(\{p_5, p_6\})$ share 
  the $k$-edge $e$. The set $Q= \{p_4, p_5\}\cup \{p_5, p_6\} = \{p_4, p_5,
  p_6\}$ gives a non-empty $(k+1)$-cell (shown hashed) which contains $e$. 
  (b) The $(k+1)$-cell $C_{2,3,5}=C^{k+1}(\{p_2, p_3, p_5\})$ is shown in gray. 
  Inside $C_{2,3,5}$, the edges of $\VD^k(P)$ are identical to the 
  edges of $\FVD(\{p_2, p_3, p_5\})$. These edges meet the 
  boundary of $C_{2,3,5}$ only in the vertices of $C_{2,3,5}$.}
\end{figure}

\noindent
We will need several well-known properties of higher-order 
Voronoi diagrams~\cite{Lee82}:
\renewcommand{\labelenumi}{(\Roman{enumi})}
\begin{enumerate}
    \itemsep0em
	\item let $Q_1, Q_2 \subset P$ be two $k$-subsets such that
the $k$-cells $C^{k}(Q_1)$ and $C^{k}(Q_2)$ are 
non-empty and adjacent (i.e., share a $k$-edge $e$). 
Then, the set $Q = Q_1 \cup Q_2$ has size $k+1$, and $C^{k+1}(Q)$ 
is a non-empty $(k+1)$-cell; see Figure~\ref{fig:5a}.
	\item Let $Q \subset P$ be a $(k+1)$-subset with $C^{k+1}(Q)$ 
non-empty. Then, the part of  $\VD^{k}(P)$ restricted to 
$C^{k+1}(Q)$ is identical to (i.e., has the same vertices and 
edges as) the part of $\FVD(Q)$ restricted to $C^{k+1}(Q)$.
Furthermore, the edges of $\FVD(Q)$ in $C^{k+1}(Q)$ do not intersect 
the boundary, but their endpoints either lie in the interior of 
$C^{k+1}(Q)$ or coincide with vertices of $C^{k+1}(Q)$.
Hence, for every $(k+1)$-cell $C$, the number of $k$-edges 
in $C$ lies between $1$ and $O(k+1)$, and 
these edges form a tree; see Figure~\ref{fig:5b}.
	\item If $v$ is an old $k$-vertex, then it is also a new $(k-1)$-vertex,
and if $v$ is a new $k$-vertex, then it is also an old $(k+1)$-vertex.  
In particular, every vertex appears in exactly two Voronoi 
diagrams of consecutive order; see Figure~\ref{fig:6}. Note that all $1$-vertices are 
new, and all $(n-1)$-vertices are old.
\end{enumerate}

\begin{figure}[tb]
  \centering
  \subcaptionbox{\label{fig:6a}}{\includegraphics[page=1]{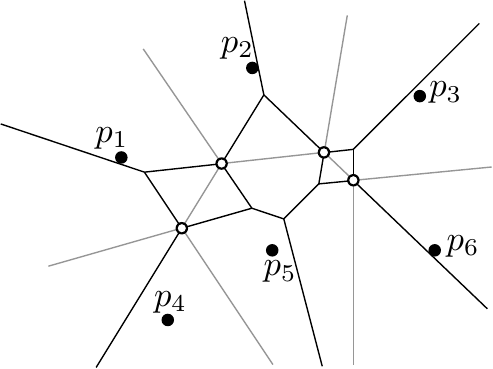}}\hspace{0.5in}
  \subcaptionbox{\label{fig:6b}}{\includegraphics[page=2]{6}}
\caption{The diagram $\VD^k(P)$ (black) for $k = 2$ and $P=\{p_1,\dots, p_6\}$. 
  (a) The diagram $\VD^{k-1}(P)$ is shown in gray. The empty  vertices of 
  $\VD^k(P)$ are old $k$-vertices, and they also appear in $\VD^{k-1}(P)$ 
  as new $(k-1)$-vertices. (b) The diagram $\VD^{k+1}(P)$ is shown in 
  gray. The empty vertices of $\VD^k(P)$ are new $k$-vertices, 
  and they also appear in $\VD^{k+1}(P)$ as old $(k+1)$-vertices. 
  Every vertex of $\VD^k(P)$ appears in exactly one of $\VD^{k-1}(P)$ or $\VD^{k+1}(P)$.}
\label{fig:6}
\end{figure}

Next, we describe a procedure to generate all (directed)
$(k+1)$-half-edges, assuming that we have all (directed) $k$-half-edges
at hand. Later, we will 
combine these procedures, for $k = 1, \dots, K$,
in a space-efficient manner.
Our high-level idea is as follows: let $e$ be a
$k$-half-edge. By property (II), the $k$-half-edge $e$ lies
inside a $(k+1)$-cell $C$. We will see that 
we can use $e$ as a starting ray to report all half-edges incident 
to $C$, similar to Lemma~\ref{le:finding s edges}.
However, if we repeat this procedure for every $k$-half-edge, we 
may report a $(k+1)$-half-edge $\Omega(k)$ times. 
This will lead to problems when we combine the procedures
for computing the Voronoi diagrams of different orders.
To avoid this, we do the following: we call a $k$-half-edge
\emph{relevant} if its head vertex lies
on the  boundary of the
$(k+1)$-cell $C$ that contains it.
For each $(k+1)$-cell $C$, we partition the boundary
of $C$ into \emph{intervals} of $(k+1)$-half-edges 
between two consecutive head vertices of relevant $k$-half-edges that 
lie inside $C$. 
We assign each such interval to
the relevant $k$-half-edge of its clockwise endpoint; see Figures~\ref{fig:7a} and~\ref{fig:7b}. 

\begin{figure}[tb]
  \centering
  \subcaptionbox{\label{fig:7a}}{\includegraphics[page=1]{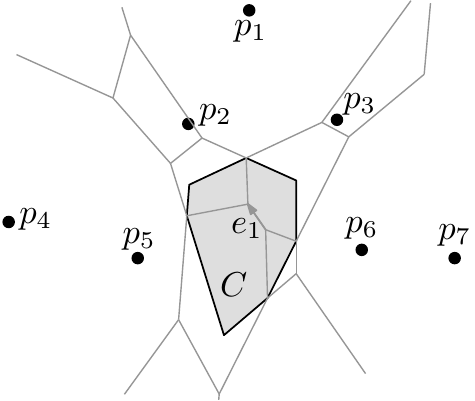}}\hspace{0.55in}
  \subcaptionbox{\label{fig:7b}}{\includegraphics[page=2]{7}}\hspace{0.7in}
  \subcaptionbox{\label{fig:7c}}{\includegraphics[page=3]{7}}
\caption{(a) The diagram $\VD^k(P)$ (gray) for $k = 3$ and 
  $P = \{p_1, \dots, p_7\}$. The $k$-half-edge $e_1$ lies in the 
  $(k+1)$-cell $C=C^{k+1}(\{p_2, p_3, p_5, p_6\})$. The head vertex 
  of $e_1$ is not on $\partial C$, thus $e_1$ is not a relevant 
  $k$-half-edge. The opposite direction of $e_1$ is also not relevant. 
  (b) The $k$-half-edges $e_2, e_3, e_4, e_5$ are relevant, 
  since their head vertices lie on $\partial C$. 
  The interval of $(k+1)$-half-edges on $C$ assigned to each of 
  these $k$-half-edges is shown. In this example, the opposite direction 
  of none of $e_2, e_3, e_4, e_5$ is relevant. (c) The $(k+1)$-half-edge 
  $f_2$ is incident to the head vertex of $e_2$ and lies to the left 
  of the directed line spanned by $e_2$. Among all such edges, $f_2$ 
  makes the smallest angle $\theta$ with $e_2$.}
\label{fig:7}
\end{figure}

Now, our algorithm goes through all $k$-half-edges. If the
current $k$-half-edge $e$ is not relevant, the  algorithm does nothing. 
Otherwise, it reports the $(k+1)$-half-edges of the interval 
assigned to $e$.
This ensures that every half-edge is reported exactly once.
As in the previous section, we distinguish between
\emph{big} and \emph{small} cells in $\VD^{k+1}(P)$,
lest we spend too much time on cells with
many incident edges. A more detailed description follows below.

The following lemma describes
an algorithm that takes
$s$ different $k$-half-edges. For each such $k$-half-edge $e$, 
the algorithm either determines that $e$ is not relevant 
or finds the first edge of the interval of $(k+1)$-half-edges 
assigned to $e$.

\begin{lemma}\label{le:s edges of k+1 using k}
Suppose we are given $s$ different $k$-half-edges $e^k_1, \dots, e_s^k$ 
represented by the subsets $E_1, \dots, E_s$ of $P$.
There is an algorithm that, for $i = 1, \dots, s$,
either determines that $e_i^k$ is not relevant,
or finds ${e}^{k+1}_i$, the first $(k+1)$-edge of 
the interval assigned to $e_i^k$.
The algorithm takes 
total expected time  
$O\big(n\log s + nk \, 2^{O(\log^* k)}\big)$ 
or total deterministic time  
$O(n\log s + nk \log k)$ and 
uses $O(sk^2)$ words of workspace.
\end{lemma}

\begin{proof}
Our algorithm proceeds analogously to Lemma~\ref{le:finding s edges}.
First, we inspect all $k$-half-edges $e_i^k$. If the head vertex $v$
of $e_i^k$ is an old $k$-vertex, then
$v$ is not a vertex of $\VD^{k+1}(P)$, and it lies in the interior 
of a $(k+1)$-cell, so $e_i^k$ is not relevant. Otherwise, $v$
is a new $k$-vertex and an old $(k+1)$-vertex, so it appears on
the boundary of a $(k+1)$-cell. In this case, we need to
determine the first $(k+1)$-half-edge for the interval assigned to $e_i^k$. 
Let $I$ be the set of all indices $i$ such that $e_i^k$
is relevant.

To determine the first half-edge of each interval, we process the
sites in $P$ in batches of size $sk$. In each iteration, 
we pick a new batch $Q$ of $sk$ sites. Then, we construct 
$\VD^{k+1}\bigl(\bigcup_{i \in I}E_i\, \cup \, Q\bigr)$ 
in $O\big(sk \log (sk) + sk^2 \, 2^{O(\log^* k)}\big)$ expected time 
or in $O(sk \log (sk) + sk^2\log k)$ deterministic time
(note 
that $\bigcup_{i \in I}E_i\, \cup \, Q$ contains  $O(sk)$ 
sites, so the diagram $\VD^{k+1}\bigl(\bigcup_{i \in I}E_i\, \cup \, Q\bigr)$
has complexity $O(sk^2)$)~\cite{chan2000random,ChanTs16}. 
By construction, the head vertex
of each $e_i^k$ with $i \in I$ belongs to the resulting diagram, 
and we can find each head vertex in $O(\log(sk^2)) = O(\log (sk))$ time by using 
a point location structure~\cite{bcko-cgaa-08}. 
Thus, we iterate over all batches, and for each $e^k_i$,
we determine the edge $f^{k+1}_i$ that appears in one of
the resulting diagrams such that (i) $f^{k+1}_i$ is incident to
the head vertex of $e^k_i$; (ii) $f^{k+1}_i$ is to the left
of the directed line spanned by $e^k_i$; and (iii) among
all such edges, $f^{k+1}_i$ makes the smallest angle with $e^k_i$; see Figure~\ref{fig:7c}.
We need $O(n/sk)$ iterations to find $f^{k+1}_i$.
Now, for each $i \in I$, the desired $(k+1)$-half-edge
$e^{k+1}_i$ is a subset of $f^{k+1}_i$. This is because, by property (I) there is
one site which is different in the second cell incident to $e^{k+1}_i$, and
this site exists in one of the 
batches. Thus, to find the other endpoint of $e_i^{k+1}$, as in 
Lemma~\ref{le:finding s edges}, we perform a second scan 
over $P$ in batches of $sk$ sites. As before, for each batch $Q$, 
we construct $\VD^{k+1}\bigl(\bigcup_{i \in I}E_i\, \cup \, Q\bigr)$ and 
we check, for each $i\in I$, where $f_i^{k+1}$ is cut-off in the new 
diagram. 
After scanning all the sites of $P$, we have the desired endpoint
of $e^{k+1}_i$. This is because the endpoint of $e^{k+1}_i$ is defined by 
one more site of $P$, and this site exists in one of the batches. 
We orient $e_i^{k+1}$ such that the
cell containing $e_i^k$ lies to the left of it.

It follows that we can process $s$ edges of $\VD^k(P)$
in $O(n/sk)$ iterations, each of which
takes $O\big(sk \log (sk) + sk^2 \, 2^{O(\log^* k)}\big)$ expected time
or $O(sk \log (sk) + sk^2\log k)$ deterministic time. Thus, we get 
$O\big(n \log s + nk \, 2^{O(\log^* k)}\big)$ total expected time or 
$O(n \log s + nk\log k)$ total deterministic time, 
using a workspace with $O(sk^2)$ 
words (for storing the intermediate Voronoi diagrams). Note that
the term $n \log(sk)$ is substituted by $n \log(s)$, since 
$n \log(sk) =  n \log s + n \log k$, and since $n \log k$ is 
dominated by $nk$ in the total running time.
\end{proof}

The algorithm from Lemma~\ref{le:s edges of k+1 using k} is
actually more general. If, instead of a $k$-half-edge $e_i^k$ that
lies inside a $(k+1)$-cell $C$,
we have a $(k+1)$-half-edge $e_i^{k+1}$  that lies on the
boundary of $C$, the same method of
processing $P$ in batches of size $sk$ allows us to find the
next $(k+1)$-half-edge incident to $C$ in counterclockwise order 
from $e_i^{k+1}$.
These two kinds of edges can be handled simultaneously.

\begin{corollary}\label{co:k_to_k+1}
Let $e_i$ denote either a $k$-half-edge or a $(k+1)$-half-edge.
Suppose we are given $s$ such half-edges $e_1, \dots, e_s$.
Then, we can find 
in total expected time $O\big(n\log s + nk \, 2^{O(\log^* k)}\big)$ or
in total deterministic time $O(n\log s + nk \log k)$ and using 
$O(sk^2)$ words of workspace a sequence $f_1, \dots, f_s$ of 
$(k+1)$-half-edges such that, for $i = 1, \dots, s$,
we have
\begin{enumerate}
\item if $e_i$ is a relevant $k$-half-edge, then  $f_i$ is the first
$(k+1)$-half-edge  of the interval for $e_i$;
\item if $e_i$ is a $k$-half-edge that is not relevant, then 
$f_i$ is null;
\item
   if $e_i$ is a $(k+1)$-half-edge, then $f_i$ is 
   the counterclockwise successor of $e_i$.
\end{enumerate}
\end{corollary}

\begin{lemma}\label{induction}
Using two scans over all $k$-half-edges, we can report all 
$(k+1)$-half-edges 
in batches of size at most $s$ such that each $(k+1)$-half-edge is 
reported exactly once.
This takes 
$O\bigl(\frac{n^2 k}{s}(\log s + k \, 2^{O(\log^* k)})\bigr)$ expected time or
$O\bigl(\frac{n^2 k}{s}(\log s + k\log k)\bigr)$ deterministic time 
using $O(sk^2)$ words of workspace.
\end{lemma}

\begin{proof}
The algorithm consists of three phases analogous of the ones introduced in Section~\ref{sec_trade}:
in the first phase, we aim at finding the big cells. 
Let $e_i$ denote either a $k$-half-edge or a $(k+1)$-half-edge. 
To find
the big cells we keep $s$ such half-edges $e_1, \dots, e_s$ in memory.
At the beginning of this phase, $e_1, \dots, e_s$ are all $k$-half-edges.
In each iteration, we apply
Corollary~\ref{co:k_to_k+1} to these half-edges, to obtain
$s$ new $(k+1)$-half-edges $f_1, \dots, f_s$. Now, for each
$i = 1, \dots, s$, three cases can apply:
(i) $f_i$ is \emph{null}, i.e., $e_i$ was not relevant.
In the next iteration, we replace $e_i$ with a fresh $k$-half-edge;
(ii)/(iii) $f_i$ is not \emph{null}.  Now we need to determine 
whether $f_i$ is the
last $(k+1)$-half-edge of its interval. For this, we check
whether the head vertex of $f_i$ is an old $(k+1)$-vertex.
(ii) If $f_i$ is not the
last $(k+1)$-half-edge of its interval, i.e., if its head vertex 
is a new $(k+1)$-vertex, we set $e_i$ to  
$f_i$ for the next iteration; otherwise, (iii) we set $e_i$ to a fresh 
$k$-half-edge.
We repeat this procedure until there are no fresh $k$-half-edges
left.

The remaining $(k+1)$-half-edges in the working memory are incident to 
the \emph{big} $(k+1)$-cells. 
For each such cell, we store the \emph{center of gravity} of its
defining sites in an array $\mathcal{B}^{k+1}$, sorted according
to lexicographic order.
We emphasize that in the first phase, we do not report any 
$(k+1)$-half-edge.

In the second phase, we repeat the same procedure as in the first phase, 
but now that we know the big $(k+1)$-cells, we can report edges.
In order to avoid repetitions, we only report (i) every $(k+1)$-half-edge 
incident to a
small $(k+1)$-cell; and (ii) the opposite direction of every 
$(k+1)$-half-edge $e$ incident to a
small $(k+1)$-cell, so that the $(k+1)$-cell on the right of $e$ is a 
big $(k+1)$-cell.
We use $\mathcal{B}^{k+1}$ to identify the big cells, by
locating the center of gravity of the defining sites of a cell
in $\mathcal{B}^{k+1}$ with a binary search, see below for details.

In the third phase, we report every $(k+1)$-half-edge $e$ that is incident 
to a big $(k+1)$-cell, while the $(k+1)$-cell on the right of $e$
is also a big $(k+1)$-cell.
Let $\{\mathcal{B}^{k+1}\}$ denote the sites that define the big $(k+1)$-cells.
We construct $\VD^{k+1}(\{\mathcal{B}^{k+1}\})$ in the working memory.
Then, we go through the sites in $P$ in batches of size $sk$,
adding the sites of each batch to $\VD^{k+1}(\{\mathcal{B}^{k+1}\})$. While 
doing this, as in the
algorithm for Lemma~\ref{T:nvd_fvd_ts}, we keep track
of how the edges of $\VD^{k+1}(\{\mathcal{B}^{k+1}\})$ are cut by the corresponding cell in the new diagrams. 
In the end, we report all $(k+1)$-edges of $\VD^{k+1}(\{\mathcal{B}^{k+1}\})$
that are not empty. By \emph{report}, we mean report two $(k+1)$-half-edges
in opposing directions. As we explained in the algorithm for Lemma~\ref{T:nvd_fvd_ts},
these $(k+1)$-half-edges cover all the $(k+1)$-half-edges incident to a big $(k+1)$-cell, 
while their right cell is also a big $(k+1)$-cell.

Regarding the running time, the first and the second phase consist of 
$O(nk/s)$ applications of Corollary~\ref{co:k_to_k+1} which takes 
$O\bigl(\frac{n^2 k}{s}(\log s + k \, 2^{O(\log^* k)})\bigr)$ total expected time 
or
$O\bigl(\frac{n^2 k}{s}(\log s + k  \log k)\bigr)$ total deterministic time. 
Creating the array $\mathcal{B}^{k+1}$ to represent the big cells 
takes $O(sk + s \log s)$ 
steps: 
we compute the center of gravity of the defining sites for 
each big $(k+1)$-cell in $O(k)$ steps. Then we sort these 
center points in lexicographcic order 
in $O(s\log s)$ steps.
A query in $\mathcal{B}^{k+1}$ takes $O(k + \log{s})$ time: 
given a query $(k+1)$-cell $C$, we compute the center of gravity for
its defining sites in $O(k)$ time. Then we use binary-search 
in $\mathcal{B}^{k+1}$ to find a big $(k+1)$-cell with the same 
center of gravity.
Aurenhammer~\cite{Aurenhammer90} showed that these centers are
pairwise distinct, so that a $(k+1)$-cell can be uniquely identified
by the center of gravity of its defining sites.\footnote{To be 
precise, Aurenhammer~\cite[Theorem~1]{Aurenhammer90} showed the 
following: take the standard lifting of $P$ onto the unit paraboloid
and compute the center of gravity for each subset of $k + 1$ lifted 
points. Call the resulting point set $R$. Then, the vertical 
projection of the lower convex hull of $R$ is dual to $\VD^{k + 1}(P)$.
In particular, the vertices of the projection are the 
centers of gravity of the defining sites for the  cells of 
$\VD^{k + 1}(P)$. Therefore, they must be pairwise distinct: 
otherwise, they could not all appear on the lower convex hull.} 

The algorithm performs at most two queries in $\mathcal{B}^{k+1}$ per 
$(k+1)$-half-edge, for a total of $O(nk)$ edges. Thus, the total 
time for the queries is $O(nk^2 + nk\log s)$.
In the third phase, constructing a $(k+1)$-order 
Voronoi diagram of $O(sk)$ sites takes 
$O(sk\log s + sk^2 \, 2^{O(\log^* k)})$ expected time or
$O(sk\log s + sk^2\log k)$ deterministic time.
We repeat it $O(n/sk)$ times, which takes 
$O(n\log s + nk \, 2^{O(\log^* k)})$ expected time or
$O(n\log s + nk\log k)$ deterministic time 
in total.
 
Overall, the running time of the algorithm simplifies to
$O\bigl(\frac{n^2 k}{s}(\log s + k \, 2^{O(\log^* k)})\bigr)$ expected time 
or
$O\bigl(\frac{n^2 k}{s}(\log s + k\log k)\bigr)$ deterministic time.
The algorithm 
uses a workspace of $O(sk^2)$
words, for running Corollary~\ref{co:k_to_k+1}, for storing big 
$(k+1)$-cells and for constructing 
Voronoi diagrams with $O(sk)$ sites.
\end{proof}

Now, in order to find the $k$-half-edges for 
all $k=1, \dots, K$, we proceed as follows:
For a parameter $s'$ (that we will define later), we compute $s'$
different $1$-edges (we report every $1$-edge 
as two $1$-half-edges in opposing directions). Then, we apply 
Lemma~\ref{induction} (with parameter $s'$) in a pipelined fashion to obtain the 
$k$-half-edges for $k = 2, \dots, K$.
In each iteration, the algorithm from Lemma~\ref{induction} consumes
at most $s'$ different $k$-half-edges from the previous order and 
produces at most
$2s'$ new $(k+1)$-half-edges to be used at the next order.
This means that if we have 
between $s'$ and $3s'$ new $k$-half-edges available in a 
buffer, then we can use them 
one by one whenever the algorithm for
computing $(k + 1)$-half-edges in Lemma~\ref{induction} requires
such a new $k$-half-edge. Whenever the size of a buffer falls below
$s'$, we run the algorithm for the previous order
until the buffer size is again between $s'$ and $3s'$.
Applying this idea for all the orders $k = 1, \dots, K-1$,  
we need to store $K-1$ buffers, each containing
up to $3s'$ half-edges for the corresponding diagram.
Since a $k$-half-edge is represented by $O(k)$ sites
from $P$, the buffer for $k$-edges requires $O(s'k)$ words
of workspace. We call this the \emph{output buffer} and denote it by $\mathcal{O}^k$. 
Furthermore, for each $k$, we need to store
$O(s')$ half-edges that reflect the current 
state of the corresponding algorithm. 
This requires $O(s'k)$ words of workspace.
This is called the \emph{private workspace} and is 
denoted by $\mathcal{P}^{k}$. 
Finally, for the algorithm that 
is currently active, we need $O(s'k^2)$ words of workspace 
to compute the Voronoi diagram of order $k$ for the next batch
of $O(s'k)$ sites from $P$ (see Lemma~\ref{induction}). Since
this workspace is used by all the algorithms, it is
called the \emph{common workspace} and denoted by $\mathcal{C}$, see below.

\begin{figure}
  \begin{center}
    \includegraphics{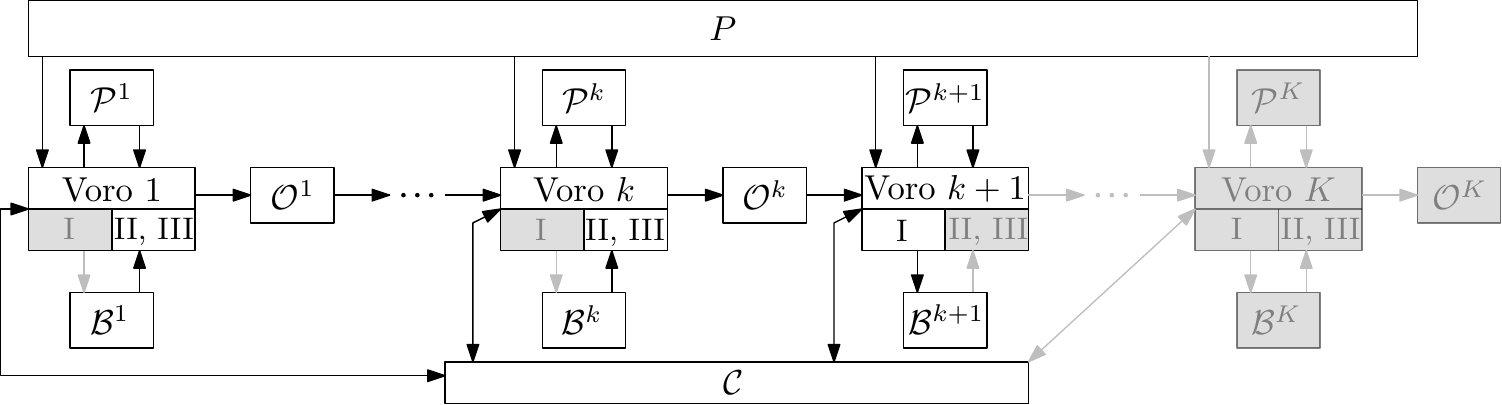}
  \end{center}
\caption{For $k'= 1,\dots, K$, Voro $k'$ is the processor 
for computing the $k'$-half-edges.
The roman numerals I, II and III refer to the 
first, second, and third phase of Voro $k'$. 
The memory cells 
$\mathcal{P}^{k'}$, $\mathcal{O}^{k'}$ and $\mathcal{B}^{k'}$
indicate the private workspace for Voro $k'$, the
output buffer for $k'$-edges, and the array for big $k'$-cells. 
The common memory of all the processors is called $\mathcal{C}$.
The figure shows the algorithm in \emph{stage} $k$.
The direction of the arrows indicates reading from or writing to 
memory cells.
The gray boxes and arrows show the inactive parts in stage $k$.
In stage $k$, the algorithm reads data from
$\mathcal{B}^1, \dots, \mathcal{B}^k$ and writes into $\mathcal{B}^{k+1}$.
In this stage, all the $k$-half-edges are reported and the 
big $(k+1)$-cells
are identified.}
\label{fig:8}
\end{figure}

\begin{theorem}
Let $P = \{p_1, \dots, p_n\}$ be a planar 
$n$-point set in general position, given in a read-only array.
Let $s$ be a parameter in $\{1, \dots, n \}$ and $K\in O(\sqrt{s})$.
We can report all the edges of 
$\VD^1(P), \dots, \VD^K(P)$
in $O\big(\frac{n^2 K^5}{s} (\log s + K \, 2^{O(\log^* K)}) \big)$ 
expected time or
in $O\big(\frac{n^2 K^5}{s} (\log s + K\log K) \big)$ 
deterministic time, using a workspace of size $O(s)$.
\end{theorem}

\begin{proof}
We compute the half-edges of $\VD^1(P), \dots, \VD^K(P)$ 
in a pipelined fashion. The algorithm simulates having $K$ processors, 
each one computing a Voronoi diagram of different order. 
For $k= 1,\dots, K$, let Voro $k$ be the processor in charge of 
computing the Voronoi diagram of order $k$. We emphasize that the 
algorithm is sequential, but the analogy of $K$ 
processors helps our exposition. Set $s' = s/K^2$. 
The first processor Voro $1$ uses the algorithm of 
Theorem~\ref{T:nvd_fvd_ts}
with space parameter $s'$. For $k \geq 2$, Voro $k$
runs the algorithm from Lemma~\ref{induction}
to compute the $k$-half-edges with space parameter $s'$. Recall 
that Lemma~\ref{induction} requires $O(s'k^2)$ words of workspace. 
This space is needed for computing $\VD^{k}(P)$ for a set of $O(s'k)$ sites. 
However, when Voro~$k$ does not compute a diagram, it 
needs only a state of $O(s'k)$ words. 

Thus, all the processors share a common workspace $\mathcal{C}$ of size $O(s'k)$. 
At any point in time, $\mathcal{C}$ is used by 
a single processor Voro $k$ to compute $\VD^{k}(P)$ (for some 
$k \in \{1, \dots, K\}$). The local state and the other variables 
needed by each processor Voro $k$ are stored in a private workspace $\mathcal{P}^k$. 
In addition, Voro $k$ has an array $\mathcal{B}^k$ to store the big $k$-cells. Whenever an edge of 
$\VD^k(P)$ (for $k \in \{1, \dots, K\}$) would be reported, we instead 
insert it into an output buffer $\mathcal{O}^{k}$. Each of these 
local arrays
should be able to store $O(s')$ half-edges and cells of $\VD^k(P)$. 
Since we need $O(k)$ sites to represent a $k$-half-edge or a $k$-cell,
the total space requirement for all processors is $O(s'k^2)=O(s)$.
 
We simulate the parallel execution of the processors with \emph{stages}.
In stage~$0$, we perform only the first phase of Theorem~\ref{T:nvd_fvd_ts}, 
to find the $O(s')$ big cells of $\VD^1(P)$, and we store them in $\mathcal{B}^1$. 
Now, we know the big $1$-cells.
Then, in stage~$1$, we perform the second and the third phase of 
Theorem~\ref{T:nvd_fvd_ts} 
to find and report the half-edges of $\VD^1(P)$ in batches of size 
at most $2s'$.
When we find a batch of $1$-half-edges, 
we store them in $\mathcal{O}^1$. 
Whenever we have at least $s'$ half-edges 
in $\mathcal{O}^1$, we pause Voro~$1$, and we start Voro~$2$
to perform the first phase of Lemma~\ref{induction}
with $\mathcal{O}^1$ as input.
This gives the half-edges of $\VD^2(P)$. 
Whenever Voro~$2$ requires new $1$-half-edges,
and the buffer $\mathcal{O}^1$ falls below $s'$ half-edges,
we continue running Voro~$1$.
When Voro~$2$  has consumed all $1$-half-edges and 
there are less than $s'$ half-edges 
in $\mathcal{P}^2$, we stop Voro~$2$ (this is the end of the first phase of Lemma~\ref{induction}). 
The current half-edges in $\mathcal{P}^2$ represent the big cells of $\VD^2(P)$, 
and we store them in $\mathcal{B}^2$. This concludes the description 
of stage~$1$. 

In general, in stage $k$ of the algorithm, we have identified the big
cells $\mathcal{B}^1, \dots, \mathcal{B}^k$ of the first $k$ diagrams, and we want to use 
Voro~$k+1$ to identify the big cells of $\VD^{k+1}(P)$. For this, 
we perform the second and the third phase of Theorem~\ref{T:nvd_fvd_ts} 
and Lemma~\ref{induction}, for
all orders $1, \dots, k$,
in a pipelined fashion to generate all half-edges of 
$\VD^1(P), \dots, \VD^k(P)$,
and we store them in the buffers $\mathcal{O}^1, \dots, \mathcal{O}^{k}$.
We also use $\mathcal{O}^{k}$ as an input 
of the first phase of Lemma~\ref{induction}, 
which gives us $\mathcal{B}^{k+1}$ for the next stage; see Figure~\ref{fig:8}.
Stage $K$ is similar, but
we do not need to determine the big cells of order $K+1$.

By running the $K$ stages of the algorithm, we compute all the Voronoi 
half-edges and add them to the corresponding output buffers.  
The edges are computed more than once. Therefore, in order to make 
sure that they are written into the output memory 
only once, we report them only the first time they are inserted into the output buffers.
For the half-edges of $\VD^{k}(P)$, this happens
in stage $k$ of the algorithm. 
Thus, we can be certain that every half-edge 
of each diagram $\VD^1(P), \dots, \VD^K(P)$ is reported exactly once and in 
order or their containing diagrams (in other words, the $k$-half-edges 
are reported before the $(k+1)$-half-edges).

Regarding the running time, 
in each stage $k = 1, \dots, K$, we have to compute all diagrams
$\VD^1(P), \dots, \VD^k(P)$, using Lemma~\ref{induction}. 
This takes 
\[
\sum_{k'=1}^{k} O\Big(\frac{n^2 k'}{s'}\big(\log s' + k' \, 2^{O(\log^*{k'})}\big)\Big) = 
O\Big(\frac{n^2 k^2}{s'}\big(\log s' + k \, 2^{O(\log^* k)}\big)\Big)
\]
expected time in stage $k$. The running time for stage~$0$ is negligible.
The complete algorithm takes 
\[
  \sum_{k=1}^{K} 
O\Big(\frac{n^2 k^2}{s'}\big(\log s' + k \, 2^{O(\log^* k)}\big)\Big) = 
O\Big(\frac{n^2 K^3}{s'}\big(\log s' + K \, 2^{O(\log^* K)}\big)\Big)
\]
expected time for all stages $1$ to $K$. This is 
$O\Big(\frac{n^2 K^5}{s}\big(\log s + K \, 2^{O(\log^* K)}\big)\Big)$
in terms of $s$, since $s'=s/K^2$.
The analysis for the deterministic running time is completely 
analogous, replacing the term $2^{O(\log^* k)}$ by $\log k$.
\end{proof}

Note that our requirement that $K = O(\sqrt{s})$ was crucial in 
ensuring that the space constraints are not exceeded; we need 
$\Theta(k)$ words of workspace to store the necessary edges of 
each $\VD^{k}(P)$, for $k = 1, \dots, K - 1$, giving a total of 
$\Theta(K^2)$ words in our workspace.

\section{Conclusion}

There are several efficient algorithms that compute a 
specific higher-order Voronoi diagram without first finding the diagrams 
of lower order~\cite{ChanTs16,AgarwalBeMaSc98,Ramos99}. 
It would be interesting to extend any of them 
to obtain a general trade-off, or even an algorithm for constant 
workspace. 
For $k = 1$ and $k = n-1$, 
our running times come close to the sorting lower 
bound which says that the time-space product for 
sorting is $\Omega(n^2)$, where the space
is measured in \emph{bits}~\cite{bc-atstosgsmc-82}.
Although improvement by a logarithmic factor may be possible, 
the gap between upper and lower 
bounds is very small.

There is a much larger gap for general higher-order 
Voronoi diagrams. We are not aware of any lower bounds (beyond the 
sorting lower bound). In particular, it would be 
interesting to have a bound in terms of the order of the diagram 
(for example, show that $\Omega(n^2K^2/s)$ steps are needed to 
find the family of all Voronoi diagrams of order up to $K$ for a 
given $n$-point set using $s$ words of workspace).
Several questions remain also unsolved when looking at upper bounds. 
Even though we do not believe our algorithms to be optimal, it seems difficult to improve them drastically. Even 
in constant sized workspaces, we do not know how to 
improve over the naive running time of $O(n^4 K)$ 
that can be obtained by computing the whole arrangement and considering
each $k \in \{1, \dots, K\}$ individually.  

\paragraph{Acknowledgments.}
The authors would like to thank Luis Barba, Kolja Junginger, 
Elena Khramtcova, and Evanthia Papadopoulou for fruitful discussions 
on this topic. We would also like to thank the anonymous referees for 
their thoughtful comments and valuable hints that helped to improve this
work.

\newcommand{\SortNoop}[1]{}

\end{document}